\documentclass[preprint,12pt]{article}
\usepackage{bbm}
\usepackage{amsfonts}
\usepackage{mathrsfs}
\usepackage{amssymb}
\usepackage{amsmath}
\usepackage{amsthm}
\usepackage{color}
\usepackage{url}
\usepackage{hyperref}
\usepackage{multirow}
\usepackage{bbding}
\usepackage[ruled,vlined]{algorithm2e}
\usepackage{latexsym}
\usepackage[final]{graphicx}
\usepackage{cases}
\usepackage{amsthm}
\usepackage{appendix}

\usepackage{lastpage}
\usepackage{epsfig}
\usepackage[misc]{ifsym}
\usepackage{lineno}

\newtheorem{definition}{Definition}
\newtheorem{theorem}{Theorem}
\newtheorem{lemma}{Lemma}
\newtheorem{remark}{Remark}

\newtheorem{example}{Counterexample}

\begin{document}
\begin{center}

\textbf{\Large{Asymptotically ideal Disjunctive Hierarchical Secret Sharing Scheme with an Explicit Construction}}

\vspace{0.7cm}

Jian Ding$^{1,2}$ $\mbox{(\Letter)}$, Cheng Wang$^{2,1}$, Haifeng Yu$^{2}$, Hongju Li$^1$, Cheng Shu$^2$

\vspace{3mm}
\footnotesize{
$^1$School of Mathematics and Big Data, Chaohu University, Hefei 238024, China\\
$^2$ School of Artificial Intelligence and Big Data, Hefei University, Hefei 230000, China}

\footnotetext{\footnotesize {This research was supported in part by University Natural Science Research Project of Anhui Province under Grants 2024AH051324 and 2025AHGXZK30056, in part by Projects of Chaohu University under Grants KYQD-202220 and hxkt20250327, and in part by Anhui Provincial Teaching Innovation Team for Linear Algebra under Grant 2024cxtd147.}}

\footnotetext{\footnotesize {\noindent$\mbox{\Letter}$ Corresponding author.}}

\footnotetext{\footnotesize {E-mail addresses: dingjian\_happy@163.com (J. Ding).}}
\end{center}

\noindent\textbf{Abstract}: Disjunctive Hierarchical Secret Sharing (DHSS) scheme is a secret sharing scheme in which the set of all participants is partitioned into disjoint subsets. Each disjoint subset is said to be a level, and different levels have different degrees of trust and different thresholds. If the number of cooperating participants from a given level falls to meet its threshold, the shortfall can be compensated by participants from higher levels. Many ideal DHSS schemes have been proposed, but they often suffer from big share sizes. Conversely, existing non-ideal DHSS schemes achieve small share sizes, yet they fail to be both secure and asymptotically ideal simultaneously. In this work, we present an explicit construct of an asymptotically ideal DHSS scheme by using a polynomial, multiple linear homogeneous recurrence relations and one-way functions. Although our scheme has computational security and many public values, it has a small share size and the dealer is required polynomial time.


\section {Introduction}
Secret sharing is a method of sharing a secret among a group of participants in which each person has a share that allows any authorized subset of participants to reconstruct the secret by pooling their shares, but any unauthorized subset of participants cannot. A secret sharing scheme is said to be perfect if no information about the secret can be known for any given an unauthorized subset. The first secret sharing scheme \cite{Shamir1979,Blakley1979} realized threshold access structures, in which each authorized subset was a subset of cardinality at least a given number. A $(t,n)$-threshold scheme consists of two phases: the share generation phase and the secret reconstruction phase. In the former, the dealer divides the secret into $n$ shares $s_i, i\in\{1,2,\ldots,n\}$, and distributes each share $s_i$ to the $i$-th participant through secure channels. In the latter, any $t$ or more participants can reconstruct the secret by pooling their shares together, while any up to $(t-1)$ participants learn nothing about the secret. Clearly, $(t,n)$-threshold schemes are perfect. The information rate is an important metric for the efficiency of a secret sharing scheme, which is defined as the ratio of the Shannon entropy of the secret to the maximum Shannon entropy of the shares. A secret sharing scheme is called ideal if it is perfect and has an information rate of one. Similarly, a scheme is termed asymptotically ideal if it is asymptotically perfect and its information rate approaches one as the secret space's cardinality becomes sufficiently large.

Threshold schemes are suitable for democratic groups, in which each participant is assigned the same degree of trust. However, for most organizations, the trust assigned to a given person is directly related to his/her position in the organizations. Disjunctive Hierarchical Secret Sharing (DHSS) schemes are suitable for this scenario. The set of all participants of a DHSS scheme are partitioned into $m$ disjoint subsets $\mathcal{P}_1,\mathcal{P}_2,\ldots,\mathcal{P}_m$. Each disjoint subset is called a level, and different levels have different degrees of trust and different thresholds. Let $t_{\ell}$ be the threshold of the $\ell$-th level for $\ell\in \{1,2,\ldots,m\}$. The secret can be reconstructed if and only if there is a number $\ell\in \{1,2,\ldots,m\}$ such that the number of participants of the first $\ell$ levels is not less than the threshold $t_{\ell}$. This shows that when the number of cooperating participants from the $\ell$-th level is $r_{\ell}$ such that $r_{\ell}<t_{\ell}$, then $(t_{\ell}-r_{\ell})$ participants can be taken from higher levels $\{1,2,\ldots,\ell-1\}$. Based on the geometric construction that was presented by Blakley \cite{Blakley1979}, Simmons \cite{Simmons} constructed the first DHSS scheme with particular threshold access structure in 1988. His scheme was not ideal and demanded (potentially exponentially) many checks when assigning identities and shares to participants. Since then, how to construct DHSS schemes has become a hot topic.

Brickell \cite{Brickell1989} constructed two ideal DHSS schemes with vector spaces over finite fields. The first one had the same drawback as Simmons' approach, requiring the dealer to check the non-singularity of possibly exponentially many matrices. The second scheme avoided this issue, and its share size was $\beta\log_2 p$ bits, where $\beta=mt_m^2$ and $p>\max_{\ell\in\{1,2,\ldots,m\}}\{|\mathcal{P}_{\ell}|\}$. Based on the Birkhoff interpolation, Tassa \cite{Tassa2007} constructed an ideal DHSS scheme with a polynomial, and its share size was $\log_2 p$ bits, where
$p>\begin{pmatrix}
   n+1 \\
   t_m
\end{pmatrix}\frac{(t_m-2)(t_m-1)}{2}+t_m$, and the dealer needed polynomial time to assign identities and shares to participants when the threshold $t_m$ was small. Note that the shares generations of the DHSS schemes \cite{Brickell1989,Tassa2007} were probabilistic.

Some other DHSS schemes were given with explicit constructions, and the dealers were required polynomial time. Ghodosi et al. \cite{Ghodosi1998} used multiple polynomials to propose an ideal DHSS scheme with an explicit construction. All those polynomials had the same constant term, which was the secret. However, the number of equations equal to the number of unknowns is not a sufficient and necessary condition for the system of equations to have a unique solution. Moreover, even if a system of equations has many solutions, one of the unknowns (the secret) may be uniquely determined. Therefore, the DHSS scheme of Ghodosi et al. \cite{Ghodosi1998} was insecure, more details can be found in appendix. Lin et al. \cite{Linchanglu2009} constructed a DHSS scheme with multiple polynomials, in which some function values of all polynomials except the first polynomial were published. Moreover, the secret was represented as a vector with two components, and these two components were taken as the constant term and the coefficient of linear term of the first polynomial, respectively. As a result, the DHSS scheme of Lin et al. \cite{Linchanglu2009} was neither ideal nor asymptotically ideal. Recently, Chen et al. \cite{Chenqi2022} gave an explicit construction of an ideal DHSS scheme using polymatroids, and reduced the algebraic number $\beta$ in the second scheme of Brickell \cite{Brickell1989} to $\frac{1}{2}\sum_{\ell=1}^{m-1}t_{\ell}(t_{\ell}-1)$.

Compared with the secret sharing schemes constructed with polynomials, Chinese Reminder Theorem (CRT)-based secret sharing schemes usually have lower computational complexity. Therefore, Harn et al. \cite{Harn-Miao2014} constructed a DHSS scheme by using the CRT over integer ring for the first, which was neither ideal nor asymptotically ideal, and many hash values were published. However, Ersoy et al. \cite{Oguzhan-Ersoy2016} pointed that the DHSS scheme of Harn et al. \cite{Harn-Miao2014} was insecure, and a new asymptotically perfect DHSS scheme was proposed with many public values. It has computational security, and its information rate was smaller than $\frac{1}{2}$. This shows that the DHSS of Ersoy et al. \cite{Oguzhan-Ersoy2016} was not asymptotically ideal. Recently, Yang et al. \cite{Yangjing2024} constructed an ideal DHSS scheme by using the CRT over polynomial ring for the first. The number of public values was smaller than that of Ersoy et al. \cite{Oguzhan-Ersoy2016}, but the DHSS scheme of Yang et al. \cite{Yangjing2024} was insecure, that is because the public values were not checked in its security analysis \cite{Hongjuarx}.

\emph{Our contributions}. By using a polynomial, multiple Linear Homogeneous Recurrence (LHR) relations and one-way functions, we propose an asymptotically ideal DHSS scheme with an explicit construction. Unlike the ideal DHSS schemes with unconditional security proposed by Brickell \cite{Brickell1989}, Tassa \cite{Tassa2007} and Chen et al.\cite{Chenqi2022}, our scheme is asymptotically ideal, computationally secure and publishes many values, but it has a smaller share size (see Table 1). Compared with the DHSS schemes of Lin et al. \cite{Linchanglu2009}, Harn et al. \cite{Harn-Miao2014}, Ersoy et al. \cite{Oguzhan-Ersoy2016} and Yang et al. \cite{Yangjing2024}, our scheme is secure and asymptotically ideal at the same time.
\begin{table}[!htb]
  \centering\scriptsize
   {\bf Table 1.}  Disjunctive hierarchical secret sharing schemes, where $t_m$ is the scheme's biggest threshold, and $|\cdot|$ is the cardinality of a set.\\
  \begin{tabular}{ccccccc}
  \hline
    \multirow{2}*{Schemes}        &Explicit                &Information               &\multirow{2}*{Ideality} &\multirow{2}*{Security}     &\multirow{2}*{Share size} \\
    ~                            &constructions            &rate                      &~                       &~                            &~\\
  \hline
  \hline
  \cite{Brickell1989}             &\multirow{2}*{No}       &\multirow{2}*{$\rho=1$}   &\multirow{2}*{Yes}      &\multirow{2}*{unconditional} &{$mt_m^2\log_2 p$ bits,}\\
     (Scheme 2)                   &~                         &~                       &~                        &~                           &$p>\max_{\ell\in\{1,2,\ldots,m\}}\{|\mathcal{P}_{\ell}|\}$\\
  \hline
  \cite{Tassa2007}(when $t_m$     &\multirow{2}*{No}         &\multirow{2}*{$\rho=1$} &\multirow{2}*{Yes}      &\multirow{2}*{unconditional } &{$\log_2 p$ bits,}\\
     was small.)                  &~                         &~                        &~    &~   &$p>\begin{pmatrix} n+1 \\ t_m\end{pmatrix}\frac{(t_m-2)(t_m-1)}{2}+t_m$\\
  \hline
  \cite{Oguzhan-Ersoy2016}       &Yes                       &$\rho<\frac{1}{2}$       &No                      &computational                &{*}\\

  \hline
  \cite{Harn-Miao2014}            &Yes                      &$\rho<1$                 &No                      &No                           &{*}\\

  \hline
  \cite{Linchanglu2009}         &Yes                      &$\rho=1$                  &No                       &unconditional                &{*}\\

  \hline
  \cite{Yangjing2024}            &Yes                      &$\rho=1$                  &Yes                     &No                            &{*}\\

  \hline
  \multirow{2}*{\cite{Chenqi2022}}    &\multirow{2}*{Yes}  &\multirow{2}*{$\rho=1$}   &\multirow{2}*{Yes}           &\multirow{2}*{unconditional}
                                                                                    &{$\frac{1}{2}\sum_{\ell=1}^{m-1}t_{\ell}(t_{\ell}-1)\log_2 p$ bits,}\\
   ~                                  &~                   &~                         &~                            &~  &$p>\max_{\ell\in\{1,2,\ldots,m\}}\{|\mathcal{P}_{\ell}|\}$\\

  \hline
  Our scheme                      &Yes                      &$\rho=1$                &Asymptotical &computational                              &{$\log_2 p$ bits, $p>n$}\\
  \hline
  \end{tabular}
 \end{table}

\emph{Related works}. This work is related to the conjunctive hierarchical secret sharing scheme of Yuan et al. \cite{YuanJiaotong2022}, which was not a DHSS scheme, and many values were published. Our scheme is a DHSS scheme that is different from Yuan et al.'s scheme. Moreover, our scheme is constructed with a polynomial, multiple LHR relations and one-way functions, while the scheme of Yuan et al. is constructed with multiple LHR relations and one-way functions. As a result, the number of public values of our scheme is less than that of the scheme of Yuan et al. \cite{YuanJiaotong2022}.

\emph{Paper organization}. After some preliminaries in Section \ref{Sec: Prelim}, the results of this work are organized as follows. In Section \ref{Sec: our scheme}, we propose an asymptotically ideal DHSS scheme, along with a comprehensive security analysis. We conclude in Section \ref{Sec: conclusion}.

\section{Preliminaries}\label{Sec: Prelim}
In this section, we will introduce basic notions and results about some secret sharing schemes and linear homogeneous recurrence relations.
\subsection{Some secret sharing schemes}

\begin{definition}[Secret sharing scheme]
Let $\mathcal{P}=\{P_1,P_2,\ldots,P_n\}$ be a group of $n$ participants. Denote by $\mathcal{S}$ and $\mathcal{S}_i$ the secret space and the share space of the participant $P_i$, respectively. A secret sharing scheme of $n$ participants includes the share generation phase and the secret reconstruction phase. In the former, for a given secret from $\mathcal{S}$, the dealer selects some random string from $\mathcal{R}$, and applies the map
                      \[\mathsf{SHARE}\colon\mathcal{S}\times\mathcal{R}\mapsto\mathcal{S}_1\times\mathcal{S}_2\times\dotsb\times\mathcal{S}_n\]
to generate shares of participants from $\mathcal{P}$. In the latter, any authorized subset $\mathcal{A}\subseteq \mathcal{P}$ can reconstruct the secret with their shares and the map
                     \[\mathsf{RECON}\colon\prod_{P_i\in\mathcal{A}}\mathcal{S}_i\mapsto\mathcal{S}.\]
Any unauthorized subset cannot reconstruct the secret. A secret sharing scheme is said to be perfect if any given unauthorized subset learns no information about the secret. The set formed by all authorized subsets in the secret sharing scheme is called the access structure.
\end{definition}

We usually consider the number $i$ as the $i$-th participant $P_i$ in this work. Let $[n]=\{1,2,\ldots,n\}$ and $[n_1, n_2]=\{n_1,n_1+1,\ldots,n_2\}$, it holds that $[n]=\mathcal{P}$. Denote by $\mathbf{X}$ and $\mathbf{Y}$ be random variables, and let $\mathsf{H}(\mathbf{X})$ be the Shannon entropy of $\mathbf{X}$. Denote by $\mathsf{H}(\mathbf{X}|\mathbf{Y})$ the conditional Shannon entropy of $\mathbf{X}$ given $\mathbf{Y}$.

\begin{definition}[Information rate, \cite{Ding2023}]\label{def: Information rate}
For a secret sharing of $n$ participants, its information rate $\rho$ is defined as the ratio of the Shannon entropy of the secret to the maximum Shannon entropy of the shares, namely,
                              \[\rho=\frac{\mathsf{H}(\mathbf{S})}{\max_{i\in [n]}^{}{\mathsf{H}(\mathbf{S}_i)}},\]
where $\mathbf{S}$ and $\mathbf{S}_i$ represent the random variables for the secret and the $i$-th participant's share, respectively. When $\mathbf{S}$ and $\mathbf{S}_i$ are random and uniformly distributed in secret space $\mathcal{S}$ and share space $\mathcal{S}_i$, respectively, it holds that
                               \[\rho=\frac{\log_2|\mathcal{S}|}{\max_{i\in [n]}^{}{\log_2|\mathcal{S}_i|}},\]
where the function $|\cdot|$ denotes the cardinality of a set.
\end{definition}

\begin{definition}[Asymptotically ideal DHSS scheme, \cite{Quisquater2002}]\label{def: Asymptotically Ideal DHSS}
Assume that a set $\mathcal{P}$ of $n$ participants is partitioned into $m$ disjoint subsets $\mathcal{P}_1, \mathcal{P}_2,\ldots, \mathcal{P}_m$, namely,
        \[\mathcal{P}=\cup_{\ell=1}^{m} \mathcal{P}_{\ell},~and~\mathcal{P}_{\ell_1}\cap\mathcal{P}_{\ell_2}=\varnothing~ for~any~ 1\leq \ell_1<\ell_2\leq m.\]
 The subset $\mathcal{P}_1$ is on the highest level of hierarchy, while $\mathcal{P}_m$ is on the least privileged level. Let $\mathcal{S}$ and $\mathcal{S}_i$ be the secret space and the share space of the $i$-th participant, respectively. Let $\mathbf{S}$ and $\mathbf{S}_i$ be the random variables representing the secret and the $i$-th participant's share, respectively. For a threshold sequence $t_1, t_2, \ldots, t_m$ such that $1\leq t_1< t_2<\cdots<t_m\leq n$, the asymptotically ideal Disjunctive Hierarchical Secret Sharing (DHSS) scheme with the given threshold sequence is a secret sharing scheme
such that the following three conditions are satisfied.
\begin{itemize}
  \item Correctness. The secret can be reconstructed by any given $\mathcal{A}\in\Gamma$, where
  \[\Gamma=\{\mathcal{A} \subseteq \mathcal{P}: \exists \ell \in [m]~such~that~|\mathcal{A}\cap(\bigcup_{w=1}^{\ell}\mathcal{P}_{w})|\geq t_{\ell}\}.\]
  \item Asymptotic perfectness. For all $\epsilon_1>0$, there is a positive integer $\sigma_1$ such that for all $\mathcal{B}\notin\Gamma$ and $|\mathcal{S}|>\sigma_1$, the loss entropy
                       \[\Delta(|\mathcal{S}|)=\mathsf{H}(\mathbf{S})-\mathsf{H}(\mathbf{S}|\mathbf{V}_{\mathcal{B}})\leq \epsilon_1,\]
   where $\mathsf{H}(\mathbf{S})\neq 0$, and $\mathbf{V}_{\mathcal{B}}$ is a random variable representing the knowledge of $\mathcal{B}$.
  \item Asymptotic maximum information rate. For all $\epsilon_2>0$, there is a positive integer $\sigma_2$ such that for all $\mathcal{B}\notin\Gamma$ and $|\mathcal{S}|>\sigma_2$, it holds that
      \[\frac{\max_{i\in [n]}^{}{\mathsf{H}(\mathbf{S}_i)}}{\mathsf{H}(\mathbf{S})}\leq 1+\epsilon_2.\]
\end{itemize}
\end{definition}

\subsection{Linear homogeneous recurrence relations}
In this subsection, we will introduce the basic notion and related results of linear recurrence recurrence relations. More details are referred to \cite{YuanJiaotong2022,Biggs1989}.

\begin{definition}[Linear homogeneous recurrence relation, \cite{Biggs1989}]\label{def: LHR}
 For given constants $a_1,a_2,\\\dots,a_{t}\in \mathbb{F}_p$ and $t$ initial values $b_0,b_1,\dots,b_{t-1}\in \mathbb{F}_p$, a Linear Homogeneous Recurrence\emph{(LHR)} relation $(u_i)_{i\geq 0}$ over $\mathbb{F}_p$ is defined by the following system of equations:
 \begin{equation*}
	\begin{cases}
		u_{0}=b_0,u_{1}=b_{1},\dots,u_{t-1}=b_{t-1}, \\
		u_{i+t}+a_{1}	u_{i+t-1}+\cdots+a_{t}	u_{i}\equiv 0\pmod p,i\geq 0.
	\end{cases}
\end{equation*}
where the positive integer $t$ is called the order of the LHR relation.
\end{definition}

\begin{definition}[Auxiliary equation,\cite{Biggs1989}]\label{def: Auxiliary equation} The auxiliary equation of the $t$-order LHR relation $(u_i)_{i\geq 0}$ in Definition \ref{def: LHR} is defined as
	\[g(x)=x^t+a_1x^{t-1}+\cdots+a_t\equiv0 \pmod p.\]
\end{definition}

\begin{theorem}[\cite{YuanJiaotong2022}]\label{Lemma: general term}
 Let $\alpha_1,\alpha_2,\dots,\alpha_m$ be distinct roots of the auxiliary equation in Definition \ref{def: Auxiliary equation} with multiplicities $t_1,t_2,\dots,t_m$, respectively. If $\sum_{\ell=1}^{m}t_{\ell}=t$, the general term for the $t$-order LHR relation $(u_i)_{i\geq 0}$ in Definition \ref{def: LHR} is
	                     \[u_i=g_1(i)\alpha_1^{i}+g_2(i)\alpha_2^{i}+\cdots+g_m(i)\alpha_m^{i}\in \mathbb{F}_p,\]
where polynomials $g_{\ell}(x)=b_{\ell,0}+b_{\ell,1}x+\cdots+b_{\ell, t_{\ell}-1}x^{t_{\ell}-1}\in \mathbb{F}_p[x], \ell\in [m]$ are determined by the $t$ initial values of $(u_i)_{i\geq 0}$. As a special case, if $\alpha_1=\alpha_2=\cdots=\alpha_m=\alpha$, the general term is
	                       \[u_i=f(i)\alpha^{i}\in \mathbb{F}_p~\mathrm{for}~f(x)=b_{1,0}+b_{1,1}x+\cdots+b_{1,t-1}x^{t-1}\in \mathbb{F}_p[x].\]
\end{theorem}

\section{An asymptotically ideal DHSS scheme}\label{Sec: our scheme}
In this section, we will construct an asymptotically ideal DHSS scheme by using a polynomial, multiple LHR relations and one-way functions. Note that all the computations in this section are performed over $\mathbb{F}_p$.

\subsection{Our scheme}
Let $\mathcal{P}$ be a set of $n$ participants, and it is partitioned into $m$ disjoint subsets $\mathcal{P}_1, \mathcal{P}_2,\ldots, \mathcal{P}_m$.
Denote by $n_{\ell}=|\mathcal{P}_{\ell}|$ and $N_{\ell}=\sum_{w=1}^{\ell} n_{w}$ for $\ell\in [m]$. Let $t_1, t_2, \ldots, t_m$ be a threshold sequence such that $1\leq t_1<t_2<\cdots<t_m$ and $t_{\ell}\leq n_{\ell}$ for $\ell\in [m]$. Our scheme consists of a share generation phase and a secret reconstruction phase.

$\bullet$ \emph{Share Generation Phase}: Let $s\in\mathbb{F}_p$ be a secret, where $p$ is a big prime such that $p>n$.

    \emph{Step 1}. The dealer selects $N_{m-1}$ random values $c_{1}, c_{2},\ldots, c_{N_{m-1}}\in \mathbb{F}_p$, and selects a random polynomial
                                \[f(x)=d_{m,1}+d_{m,2}x+\cdots+d_{m,t_{m}} x^{t_{m}-1} \in \mathbb{F}_p[x]\]
       such that $s=f(0)$. The dealer computes $\{f(i)\in \mathbb{F}_p: i\in [n]\}$, and distributes each share $s_i$ to the $i$-th participant for $i\in [n]$, where
       \[s_i=\begin{cases}
		c_i,~\mathrm{if}~ i\in [N_{m-1}], \\
		f(i),~\mathrm{if}~ i\in[N_{m-1}+1, n].
	\end{cases}\]

    \emph{Step 2}. The dealer selects $m$ publicly known distinct one-way functions $h_1(\cdot),h_2(\cdot),\dots,\\h_m(\cdot)$ that map elements of $\mathbb{F}_p$ to $\mathbb{F}_p$. The dealer selects $(m-1)$ publicly known distinct elements $\alpha_1,\alpha_2,\dots,\alpha_{m-1}\in\mathbb{F}_p^*$, and takes \[(x-\alpha_\ell)^{t_{\ell}}=x^{t_{\ell}}+a_{\ell,1}x^{t_{\ell}-1}+a_{\ell,2}x^{t_{\ell}-2}+\cdots+a_{\ell, t_{\ell}}=0, \ell\in [m-1]\]
    as auxiliary equations to constructs $(m-1)$ different LHR relations as follows:
         \[\begin{cases}
		    u_{0}^{(\ell)}=h_{\ell}(s_1),u_{1}^{(\ell)}=h_{\ell}(s_2),\dots,u_{t_{\ell}-1}^{(\ell)}=h_{\ell}(s_{t_{\ell}}), \\
		    u_{i+t_{\ell}}^{(\ell)}+a_{\ell,1}u_{i+t_{\ell}-1}^{(\ell)}+\cdots+a_{\ell, t_{\ell}}u_{i}^{(\ell)}=0, i\geq 0.
	      \end{cases}\]

     \emph{Step 3}. For any $\ell\in [m-1]$, the dealer computes and publishes $r_{\ell}=s-u_{n}^{(\ell)}$, and $I_i^{(\ell)}=u_{i-1}^{(\ell)}-h_{\ell}(s_i)$ for all $i\in [t_\ell+1,N_{\ell}]$.

     \emph{Step 4}. The dealer computes and publishes $I_i^{(m)}=f(i)-h_{m}(s_i)$ for $i\in [N_{m-1}]$.

\begin{itemize}
    \item Secret \emph{Reconstruction Phase}: Assume that
      \[\mathcal{A}=\{i_{1},i_{2},\dots,i_{|\mathcal{A}|}\}\subseteq \mathcal{P}, i_{1}<i_{2}<\cdots<i_{|\mathcal{A}|},i_{|\mathcal{A}|}\in \mathcal{P}_\ell~\mathrm{and}~|\mathcal{A}|\geq t_{\ell}.\]
\end{itemize}

   \emph{Case 1}. If $\ell\in [m-1]$, participants of $\mathcal{A}$ use their shares $\{s_i: i\in \mathcal{A}\}$, publicly known values $r_\ell, \alpha_{\ell}, I_{i}^{(\ell)}$ and publicly known one-way function $h_{\ell}(\cdot)$ to reconstruct the secret
 	      \[s=r_\ell+\alpha_{\ell}^{n}\sum_{i \in \mathcal{A}}\left(\frac{u_{i-1}^{(\ell)}}{\alpha_{\ell}^{i-1}}\prod_{j\in \mathcal{A},j\neq i}\frac{n-j+1)}{i-j}\right),\]
where
           \[u_{i-1}^{(\ell)}=
            \begin{cases}
		    h_{\ell}(s_i), & \mathrm{if}~i\in \mathcal{A}\cap [t_\ell],\\
		    I_{i}^{(\ell)}+h_{\ell}(s_i), & \mathrm{if}~i\in \mathcal{A}\cap [t_{\ell}+1,N_{\ell}].
	        \end{cases}\]

  \emph{Case 2}. If $\ell=m$, participants of $\mathcal{A}$ use their shares $\{s_i: i\in \mathcal{A}\}$, publicly known values $I_{i}^{(m)}$ and publicly known one-way function $h_{m}(\cdot)$ to reconstruct the secret
            \[s=\sum_{i \in \mathcal{A}}\left(f(i)\prod_{j\in \mathcal{A},j\neq i}\frac{-j}{i-j}\right),\]
where
           \[f(i)=
            \begin{cases}
		    I_{i}^{(m)}+h_{m}(s_i), & \mathrm{if}~i\in \mathcal{A}\cap [N_{m-1}],\\
		    s_i, & \mathrm{if}~i\in \mathcal{A}\cap [N_{m-1}+1,n].
	        \end{cases}\]

\subsection{Security analysis of our scheme}
It is easy to check that the information rate of our scheme is one. Next, we will prove the correctness and asymptotic perfectness of our scheme.
\begin{theorem}[Correctness]\label{Theorem:Correctness of our scheme} Any subset
       \[\mathcal{A}\in \Gamma=\{\mathcal{A} \subseteq \mathcal{P}: \exists \ell \in [m]~such~that~|\mathcal{A}\cap(\bigcup_{w=1}^{\ell}\mathcal{P}_{\ell})|\geq t_{\ell}\}\]
can reconstruct the secret by using their shares, publicly known values and publicly known one-way functions.
\end{theorem}
\begin{proof}
Without loss of generality, we let
          \[\mathcal{A}=\{i_{1},i_{2},\dots,i_{|\mathcal{A}|}\}\subseteq \mathcal{P}, i_{1}<i_{2}<\cdots<i_{{|\mathcal{A}|}},i_{|\mathcal{A}|}\in \mathcal{P}_\ell~\mathrm{and}~|\mathcal{A}|\geq t_{\ell}.\]
The proof is divided into two cases according to the range of $\ell$.

\emph{Case 1}. If $\ell\in [m-1]$, participants of $\mathcal{A}$ use their shares $\{s_i: i\in \mathcal{A}\}$, publicly known values $I_{i}^{(\ell)}$ and publicly known one-way function $h_{\ell}(\cdot)$ to get $|\mathcal{A}|$ terms $\{u_{i-1}^{(\ell)}: i\in \mathcal{A}\}$ of the $\ell$-th LHR relation $(u_i^{(\ell)})_{i\geq 0}$, where
 	      \[u_{i-1}^{(\ell)}=
            \begin{cases}
		    h_{\ell}(s_i), & \mathrm{if}~i\in \mathcal{A}\cap [t_\ell],\\
		    I_{i}^{(\ell)}+h_{\ell}(s_i), & \mathrm{if}~i\in \mathcal{A}\cap [t_{\ell}+1,N_{\ell}].
	        \end{cases}\]
Since the auxiliary equation of $(u_i^{(\ell)})_{i\geq 0}$ is
                    \[(x-\alpha_\ell)^{t_{\ell}}=x^{t_{\ell}}+a_{\ell,1}x^{t_{\ell}-1}+a_{\ell,2}x^{t_{\ell}-2}+\cdots+a_{\ell, t_{\ell}}=0,\]
from Lemma \ref{Lemma: general term} we have that the general term of $(u_i^{(\ell)})_{i\geq 0}$ is $u_i^{(\ell)}=f^{(\ell)}(i)\alpha_{\ell}^{i}$, where the degree of $f_{\ell}(x)\in \mathbb{F}_p[x]$ is less than $t_{\ell}$. This shows that participants of $\mathcal{A}$ get $\{f_{\ell}(i-1)=\frac{u_{i-1}^{(\ell)}}{\alpha_{\ell}^{i-1}}: i\in \mathcal{A}\}$ of cardinality $|\mathcal{A}|\geq t_{\ell}$. From Lagrange interpolating formula, it holds that
    \begin{displaymath}
       \begin{aligned}
         f_{\ell}(x)=&\sum_{i \in \mathcal{A}}\left(f_{\ell}(i-1)\prod_{j\in \mathcal{A},j\neq i}\frac{x-j+1)}{i-j}\right)\\
                      =&\sum_{i \in \mathcal{A}}\left(\frac{u_{i-1}^{(\ell)}}{\alpha_{\ell}^{i-1}}\prod_{j\in \mathcal{A},j\neq i}\frac{x-j+1)}{i-j}\right).
     \end{aligned}
    \end{displaymath}
As a result,
\begin{displaymath}
    \begin{aligned}
           s=&r_\ell+u_n^{(\ell)}\\
            =&r_\ell+f_{\ell}(n){\alpha_{\ell}^{n}}\\
            =&r_\ell+\alpha_{\ell}^{n}\sum_{i \in \mathcal{A}}\left(\frac{u_{i-1}^{(\ell)}}{\alpha_{\ell}^{i-1}}\prod_{j\in \mathcal{A},j\neq i}\frac{n-j+1)}{i-j}\right).
    \end{aligned}
 \end{displaymath}

\emph{Case 2}. If $\ell=m$, participants of $\mathcal{A}$ use their shares $\{s_i: i\in \mathcal{A}\}$, publicly known values $I_{i}^{(m)}$ and publicly known one-way function $h_{m}(\cdot)$ to get $\{f(i): i\in \mathcal{A}\}$ of cardinality $|\mathcal{A}|\geq t_{m}$, where
        \[f(i)=
            \begin{cases}
		    I_{i}^{(m)}+h_{m}(s_i), & \mathrm{if}~i\in \mathcal{A}\cap [N_{m-1}],\\
		    s_i, & \mathrm{if}~i\in \mathcal{A}\cap [N_{m-1}+1,n].
	        \end{cases}\]
Since the degree of $f(x)$ is less than $t_m$, from Lagrange interpolating formula we get that
            \[s=\sum_{i \in \mathcal{A}}\left(f(i)\prod_{j\in \mathcal{A},j\neq i}\frac{-j}{i-j}\right).\]
\end{proof}

Now we will prove the \emph{asymptotic perfectness} of our scheme. Assume that the subset $\mathcal{B}\subset \mathcal{P}$ is an unauthorized subset, which shows that
            \[\mathcal{B}\notin \Gamma=\{\mathcal{A} \subseteq \mathcal{P}: \exists \ell \in [m]~such~that~|\mathcal{A}\cap(\bigcup_{w=1}^{\ell}\mathcal{P}_{w})|\geq t_{\ell}\}.\]
Participants of $\mathcal{B}$ know the upper bound of the degree of $f(x)\in \mathbb{F}_p[x]$, all publicly known one-way functions and values, as well as their shares. Namely, the knowledge of $\mathcal{B}$ is the following seven conditions denoted by $\mathcal{V}_{\mathcal{B}}$.
\begin{itemize}
  \item[(i)] $\deg(f(x))<t_m$, and the equation $f(0)=s$.
  \item[(ii)] One-way functions $h_1(\cdot),h_2(\cdot),\dots,h_m(\cdot)$ are publicly known.
  \item[(iii)] Distinct elements $\alpha_1,\alpha_2,\dots,\alpha_{m-1}\in\mathbb{F}_p^*$ are publicly known, and \[(x-\alpha_\ell)^{t_{\ell}}=x^{t_{\ell}}+a_{\ell,1}x^{t_{\ell}-1}+a_{\ell,2}x^{t_{\ell}-2}+\cdots+a_{\ell, t_{\ell}}=0, \ell\in [m-1]\]
    are auxiliary equations of $(m-1)$ different LHR relations, respectively.

  \item[(iv)] The values $r_{\ell}=s-u_{n}^{(\ell)}, \ell\in [m-1]$ are publicly known.

  \item[(v)] For any $i\in \mathcal{B}\cap [N_{m-1}]$, it holds that $f(i)=I_i^{(m)}+h_{m}(s_i)$, and there is an integer $\ell_1\in [m-1]$ such that $i\in \mathcal{P}_{\ell_1}$. For all $\ell\in [\ell_1,m-1]$, it has that
            \[u_{i-1}^{(\ell)}=
            \begin{cases}
		    h_{\ell}(s_i), & \mathrm{if}~i\in \mathcal{B}\cap [t_{\ell}],\\
		    I_{i}^{(\ell)}+h_{\ell}(s_i), & \mathrm{if}~i\in \mathcal{B}\cap [t_{\ell}+1,N_{\ell}].
	        \end{cases}\]
  \item[(vi)] For $i\in \mathcal{B}\cap [N_{m-1}+1, N_m]$, it holds that $f(i)=s_i$.
  \item[(vii)] For any $i\in [N_{m-1}]$ and $i\notin \mathcal{B}$, there is an integer $\ell_2\in [m]$ such that $i\in \mathcal{P}_{\ell_2}$. There is a value $\widetilde{s}_i\in \mathbb{F}_p$ such that $h_{m}(\widetilde{s}_i)=f(i)-I_i^{(m)}$, and for all $\ell\in [\ell_2,m-1]$, it has that
            \[ h_{\ell}(\widetilde{s}_i)=
            \begin{cases}
		    u_{i-1}^{(\ell)}, &\mathrm{if}~i\notin \mathcal{B}~\mathrm{and}~i\in[t_{\ell}],\\
		    u_{i-1}^{(\ell)}-I_{i}^{(\ell)}, & \mathrm{if}~i\notin \mathcal{B}~\mathrm{and}~ i\in [t_{\ell}+1,N_{\ell}].
	        \end{cases}\]
   \end{itemize}
Note that the conditions (v) and (vii) are used to make $\{(u_i^{(\ell)})_{i\geq 0}: \ell\in [m-1]\}, f(x)$ satisfy publicly known $I_i^{(\ell)}$.

\begin{lemma}\label{le: loss entropy}
Denote by $\mathcal{V}^\prime_{\mathcal{B}}$ the conditions (i) to (vi). For all $\epsilon_1>0$, there is a positive integer $\sigma_1$ such that for all $\mathcal{B}\notin\Gamma$ and $|\mathcal{S}|=p>\sigma_1$, it holds that
                               \[0<\mathsf{H}(\mathbf{S}|\mathbf{V}^\prime_{\mathcal{B}})-\mathsf{H}(\mathbf{S}|\mathbf{V}_{\mathcal{B}})<\epsilon_1,\]
where $\mathbf{V}^\prime_{\mathcal{B}}$ and $\mathbf{V}_{\mathcal{B}}$ represent random variables of $\mathcal{V}^\prime_{\mathcal{B}}$ and $\mathcal{V}_{\mathcal{B}}$, respectively.
\end{lemma}
\begin{proof}
Since the shares $\{s_i: i\in [N_{m-1}],i\notin \mathcal{B}\}$ are randomly selected by the dealer and functions $\{h_{\ell}(\cdot): \ell\in [m]\}$ are distinct one-way functions, then the condition (vii) can eliminate a group of $\{(u_i^{(\ell)})_{i\geq 0}: \ell\in [m-1]\}, f(x)$ with a negligible probability when $|\mathcal{S}|=p$ is big enough. This shows that $\mathsf{H}(\mathbf{S}|\mathbf{V}^\prime_{\mathcal{B}})-\mathsf{H}(\mathbf{S}|\mathbf{V}_{\mathcal{B}})$ is negligible when $p$ is big enough. This gives the proof.
\end{proof}

\begin{lemma}\label{le: i-v-nosecret}
$\mathcal{V}^\prime_{\mathcal{B}}$ contains no information about the secret, namely, $\mathsf{H}(\mathbf{S}|\mathbf{V}^\prime_{\mathcal{B}})=\mathsf{H}(\mathbf{S})$.
\end{lemma}

\begin{proof}
Consider the worst case, namely,
            \[|\mathcal{B}\cap(\bigcup_{w=1}^{\ell}\mathcal{P}_{w})|=|\mathcal{B}\cap [N_{\ell}]|= t_{\ell}-1~\mathrm{for~all}~\ell\in [m],\]
and denote by $\mathcal{B}=\{i_{1},i_{2},\dots,i_{|\mathcal{B}|}\}\subset \mathcal{P}, i_{1}<i_{2}<\cdots<i_{{|\mathcal{B}|}}$.
It is easy to check that $|\mathcal{B}|=|\mathcal{B}\cap(\bigcup_{w=1}^{m}\mathcal{P}_{w})|=t_{m}-1$ and $i_{|\mathcal{B}|}\in \mathcal{P}_m$. For any secret $\widetilde{s}\in \mathbb{F}_p$, we only need to prove that there is a unique set of LHR relations $\{(\widetilde{u}_i^{(\ell)})_{i\geq 0}: \ell\in [m-1]\}$ and a unique polynomial $\widetilde{f}(x)\in \mathbb{F}_p[x]$ of degree less than $t_m$ such that the knowledge $\mathcal{V}^\prime_{\mathcal{B}}$ is still correct by replacing $s, \{(u_i^{(\ell)})_{i\geq 0}: \ell\in [m-1]\}, f(x)$ with $\widetilde{s}, \{(\widetilde{u}_i^{(\ell)})_{i\geq 0}: \ell\in [m-1]\}, \widetilde{f}(x)$, respectively. Consequently, $\mathcal{V}^\prime_{\mathcal{B}}$ contains no information about the secret. The proof is classified into three parts.

First, we will prove the existence of the unique set of LHR relations $\{(\widetilde{u}_i^{(\ell)})_{i\geq 0}: \ell\in [m-1]\}$. For any secret $\widetilde{s}\in \mathbb{F}_p$ and $\ell\in [m-1]$, participants of $\mathcal{B}$ use their shares $\{s_i: i\in \mathcal{B}\cap [N_{\ell}]\}$, publicly known values $I_{i}^{(\ell)}, r_{\ell}$ and publicly known one-way function $h_{\ell}(\cdot)$ to define $t_{\ell}$ terms $\{\widetilde{u}_{i-1}^{(\ell)}: i\in (\mathcal{B}\cap [N_{\ell}])\cup \{n+1\}\}$ as
          \[\widetilde{u}_{i-1}^{(\ell)}=
           \begin{cases}
             u_{i-1}^{(\ell)}, & \mathrm{if}~i\in \mathcal{B}\cap [N_{\ell}],\\
            \widetilde{s}-r_{\ell}, & \mathrm{if}~i=n+1.\\
           \end{cases}~~~~~~~~~~~~~~~~~~~~~~\]
          \begin{equation}\label{Equation: construction of LHR}
           =\begin{cases}
		    h_{\ell}(s_i), & \mathrm{if}~i\in \mathcal{B}\cap [t_\ell],\\
		    I_{i}^{(\ell)}+h_{\ell}(s_i), & \mathrm{if}~i\in \mathcal{B}\cap [t_{\ell}+1,N_{\ell}],\\
            \widetilde{s}-r_{\ell}, & \mathrm{if}~i=n+1.
	        \end{cases}
          \end{equation}
The auxiliary equation of $(u_i^{(\ell)})_{i\geq 0}$, i.e.,
                    \[(x-\alpha_\ell)^{t_{\ell}}=x^{t_{\ell}}+a_{\ell,1}x^{t_{\ell}-1}+a_{\ell,2}x^{t_{\ell}-2}+\cdots+a_{\ell, t_{\ell}}=0,\]
is seen as the auxiliary equation of $(\widetilde{u}_i^{(\ell)})_{i\geq 0}$. By a similar discussion as the Case 1 in the proof of Theorem \ref{Theorem:Correctness of our scheme}, the LHR relation $(\widetilde{u}_i^{(\ell)})_{i\geq 0}$ is unique, and its general term is $\widetilde{u}_i^{(\ell)}=\widetilde{f}_{\ell}(i)\alpha_{\ell}^{i}$, where
                    \[\widetilde{f}_{\ell}(x)=\sum_{i \in \mathcal{B}\cup\{n+1\}}\left(\frac{u_{i-1}^{(\ell)}}{\alpha_{\ell}^{i-1}}\prod_{j\in \mathcal{B}\cup\{n+1\},j\neq i}\frac{x-j+1)}{i-j}\right).\]

Second,  we will prove the existence of the unique polynomial $\widetilde{f}(x)\in \mathbb{F}_p[x]$. For the secret $\widetilde{s}\in \mathbb{F}_p$, participants of $\mathcal{B}$ use their shares $\{s_i: i\in \mathcal{B}\}$, publicly known values $I_{i}^{(m)}$ and publicly known one-way function $h_{m}(\cdot)$ to define $\{\widetilde{f}(i): i\in \mathcal{B}\cup\{0\}\}$ of cardinality $|\mathcal{B}\cup\{0\}|=t_{m}$, where
        \[\widetilde{f}(i)=\begin{cases}
		    f(i), & \mathrm{if}~i\in \mathcal{B},\\
		    \widetilde{s}, & \mathrm{if}~i=0.\\
	        \end{cases}~~~~~~~~~~~~~~~~~~~~~~~~~~~~~~~~~\]
        \begin{equation}\label{Equation:construction of f(x)}
           =\begin{cases}
		    I_{i}^{(m)}+h_{m}(s_i), & \mathrm{if}~i\in \mathcal{B}\cap [N_{m-1}],\\
		    s_i, & \mathrm{if}~i\in \mathcal{B}\cap [N_{m-1}+1,n],\\
            \widetilde{s}, & \mathrm{if}~i=0.
	        \end{cases}
          \end{equation}
From Lagrange interpolating formula, we get a unique
            \begin{equation}\label{Equation: degree of f(x)}
            \widetilde{f}(x)=\sum_{i \in \mathcal{B}\cup\{0\}}\left(f(i)\prod_{j\in \mathcal{B}\cup\{0\},j\neq i}\frac{x-j}{i-j}\right),
            \end{equation}
where the degree of $\widetilde{f}(x)$ is less than $|\mathcal{B}\cup\{0\}|=t_{m}$, and $\widetilde{f}(0)=\widetilde{s}$.

Third, we will prove that the knowledge $\mathcal{V}^\prime_{\mathcal{B}}$ is still correct by replacing $s, \{(u_i^{(\ell)})_{i\geq 0}: \ell\in [m-1]\}, f(x)$ with $\widetilde{s}, \{(\widetilde{u}_i^{(\ell)})_{i\geq 0}: \ell\in [m-1]\}, \widetilde{f}(x)$, respectively.
\begin{itemize}
  \item From Equations (\ref{Equation:construction of f(x)}) and (\ref{Equation: degree of f(x)}), it holds that $\deg(\widetilde{f}(x))<t_m$, and $\widetilde{f}(0)=\widetilde{s}$.

  \item  The auxiliary equations of $\{(\widetilde{u}_i^{(\ell)})_{i\geq 0}: \ell\in [m-1]\}$ are
  \[(x-\alpha_\ell)^{t_{\ell}}=x^{t_{\ell}}+a_{\ell,1}x^{t_{\ell}-1}+a_{\ell,2}x^{t_{\ell}-2}+\cdots+a_{\ell, t_{\ell}}=0, \ell\in [m-1],\]
    respectively.

  \item From the constructions of $\{(\widetilde{u}_i^{(\ell)})_{i\geq 0}: \ell\in [m-1]\}$ in Equation (\ref{Equation: construction of LHR}), it holds that $r_{\ell}=\widetilde{s}-\widetilde{u}_{n}^{(\ell)}$ for all $\ell\in [m-1]$.

  \item Let $i\in \mathcal{B}\cap [N_{m-1}]$. From the construction of $\widetilde{f}(x)$ in Equation (\ref{Equation:construction of f(x)}), it gets that $\widetilde{f}(i)=I_i^{(m)}+h_{m}(s_i)$. There is an integer $\ell_1\in [m-1]$ such that $i\in \mathcal{P}_{\ell_1}$. For all $\ell\in [\ell_1,m-1]$, from the constructions of $\{(\widetilde{u}_i^{(\ell)})_{i\geq 0}: \ell\in [m-1]\}$ in Equation (\ref{Equation: construction of LHR}), it has that
            \[\widetilde{u}_{i-1}^{(\ell)}=
            \begin{cases}
		    h_{\ell}(s_i), & \mathrm{if}~i\in \mathcal{B}\cap [t_{\ell}],\\
		    I_{i}^{(\ell)}+h_{\ell}(s_i), & \mathrm{if}~i\in \mathcal{B}\cap [t_{\ell}+1,N_{\ell}].
	        \end{cases}\]
  \item From the construction of $\widetilde{f}(x)$ in Equation (\ref{Equation:construction of f(x)}), it holds that $\widetilde{f}(i)=s_i$ for $i\in \mathcal{B}\cap [N_{m-1}+1, N_m]$.
     \end{itemize}
Therefore, $\mathcal{V}^\prime_{\mathcal{B}}$ contains no information about the secret, namely, $\mathsf{H}(\mathbf{S}|\mathbf{V}^\prime_{\mathcal{B}})=\mathsf{H}(\mathbf{S})$.
  \end{proof}

\begin{theorem}[Asymptotic perfectness]\label{Th:Asymptotic perfectness}
Our scheme is asymptotically perfect.
\end{theorem}
\begin{proof}
Let $\mathcal{B}\notin \Gamma$. By Lemma \ref{le: loss entropy}, we get that for all $\epsilon_1>0$, there is a positive integer $\sigma_1$ such that for all $\mathcal{B}\notin\Gamma$ and $|\mathcal{S}|=p>\sigma_1$, it holds that
                               \[0<\mathsf{H}(\mathbf{S}|\mathbf{V}^\prime_{\mathcal{B}})-\mathsf{H}(\mathbf{S}|\mathbf{V}_{\mathcal{B}})<\epsilon_1.\]
Since $\mathsf{H}(\mathbf{S}|\mathbf{V}^\prime_{\mathcal{B}})=\mathsf{H}(\mathbf{S})$ by Lemma \ref{le: i-v-nosecret}, then we have
                              \[0<\mathsf{H}(\mathbf{S})-\mathsf{H}(\mathbf{S}|\mathbf{V}_{\mathcal{B}})<\epsilon_1.\]
This shows that our scheme is asymptotically perfect according to Definition \ref{def: Asymptotically Ideal DHSS}.
\end{proof}

\begin{theorem}\label{Theorem: summary}
Let $\mathcal{P}$ be a set of $n$ participants, and it is partitioned into $m$ disjoint subsets $\mathcal{P}_1, \mathcal{P}_2,\ldots, \mathcal{P}_m$. Denote by $n_{\ell}=|\mathcal{P}_{\ell}|$ for $\ell\in [m]$. For a threshold sequence $t_1, t_2, \ldots, t_m$ such that $1\leq t_1< t_2<\cdots<t_m\leq n$ and $t_{\ell}\leq n_{\ell}$ for $\ell\in [m]$, our scheme is an asymptotically ideal DHSS scheme with share size $\log_2 p$ bits for $p>n$, and the dealer is required polynomial time.
\end{theorem}
\begin{proof}
From Definition \ref{def: Asymptotically Ideal DHSS}, Theorem \ref{Theorem:Correctness of our scheme} and Theorem \ref{Th:Asymptotic perfectness}, we get that our scheme is an asymptotically perfect DHSS scheme. Since the secret space and every share space are the same finite field $\mathbb{F}_p$, the information rate is $\rho=1$, and the share size is $\log_2 p$ bits for $p>n$. Therefore, our scheme is asymptotically ideal. It is easy to check that the dealer of our scheme is required polynomial time.
\end{proof}

 Compared with the ideal DHSS schemes with unconditional security proposed by Brickell \cite{Brickell1989}, Tassa \cite{Tassa2007} and Chen et al.\cite{Chenqi2022}, our scheme is asymptotically ideal, computationally secure and publishes many values, but it has a smaller share size. Compared with the DHSS schemes of Lin et al. \cite{Linchanglu2009}, Harn et al. \cite{Harn-Miao2014}, Ersoy et al. \cite{Oguzhan-Ersoy2016} and Yang et al. \cite{Yangjing2024}, our scheme is secure and asymptotically ideal at the same time.

\section{Conclusion}\label{Sec: conclusion}
In this work, we construct an asymptotically ideal DHSS scheme with a small share size. The dealer of our scheme is required polynomial time, but our scheme has many public values. How to reduce the number of public values is our future work.
\section*{Appendix: Security analysis of Ghodosi et al. scheme}
\appendix
In this section, we review the scheme of Ghodosi et al.\cite{Ghodosi1998}, which is called GPN scheme in the following. The authors claimed that GPN scheme was a perfect and ideal DHSS scheme. We point that GPN scheme is insecure, and two counterexamples are given to attack the correctness and privacy of GPN scheme.

\section[\appendixname~\thesection]{Review of GPN scheme}\label{Subsec: review of GPN}
\setcounter{equation}{0}
\renewcommand{\theequation}{A\arabic{equation}}

Let $\mathcal{P}$ be a set of $n$ participants. We introduce GPN scheme with total levels $m=2$, namely, the set $\mathcal{P}$ is partitioned into $2$ disjoint subsets $\mathcal{P}_1$ and $\mathcal{P}_2$. Denote by $n_{1}=|\mathcal{P}_{1}|$ and $n_{2}=|\mathcal{P}_{2}|$, then $n=n_1+n_2$. Let $t_1, t_2$ be thresholds such that $1\leq t_1<t_2<n_2$ and $t_{1}\leq n_{1}$. The GPN scheme consists of a share generation phase and a secret reconstruction phase.

$\bullet$ \emph{Share Generation Phase}: Let $s\in\mathbb{F}_p$ be a secret, where $p$ is a big prime such that $p>n$. The dealer chooses two random polynomials
        \[f_1(x)=a_{1,0}+a_{1,1}x+\cdots+a_{1,t_1-1}x^{t_1-1}\in \mathbb{F}_p[x],\]
        \[f_2(x)=a_{2,0}+a_{2,1}x+\cdots+a_{2,n_1+t_2-t_1}x^{n_1+t_2-t_1}\in \mathbb{F}_p[x],\]
and $n$ publicly known distinct non-zero elements $x_1, x_2, \ldots, x_n\in \mathbb{F}_{p}^{*}$ such that
                          \[f_1(0)=f_2(0)=s\in \mathbb{F}_p,~\mathrm{and}~f_1(x_i)=f_2(x_i)\in \mathbb{F}_p~\mathrm{for}~i\in [n_1].\]
The dealer computes $s_i=f(x_{i})\in \mathbb{F}_p$ for $i\in [n]$, and distributes each share $s_i$ to the $i$-th participant $P_{i}$.

\begin{itemize}
    \item \emph{Secret Reconstruction Phase}: Assume that
      \[\mathcal{A}=\{i_{1},i_{2},\dots,i_{|\mathcal{A}|}\}\subseteq \mathcal{P}, i_{1}<i_{2}<\cdots<i_{|\mathcal{A}|},i_{|\mathcal{A}|}\in \mathcal{P}_\ell~\mathrm{and}~|\mathcal{A}|=t_{\ell}.\]
\end{itemize}

   \emph{Case 1}. If $\ell=1$, participants of $\mathcal{A}$ use their shares $\{s_i: i\in \mathcal{A}\}$, publicly known values $\{x_i:i\in \mathcal{A}\}$ and the Lagrange interpolating formula to reconstruct the secret
                           \[s=\sum_{i\in {\mathcal{A}}}\left(f(x_{i})\prod_{j\in {\mathcal{A}}, j\neq i}\frac{-x_{j}}{x_{i}-x_{j}}\right) \in \mathbb{F}_{p}.\]

   \emph{Case 2}. If $\ell=2$, participants of $\mathcal{A}$ use their shares $\{s_i: i\in \mathcal{A}\}$, publicly known values $\{x_i:i\in \mathcal{A}\}$ to reconstruct the secret by solving the following system of equations over $\mathbb{F}_p$.
      \begin{equation}\label{Equation:GPN reconstruction}
       \begin{cases}
		f_1(x_i)=f_2(x_i),i\in [n_1],\\
		f_2(x_i)=s_i,i\in \mathcal{A}.\\
	    \end{cases}
      \end{equation}
\begin{remark}
    In the system of equations (\ref{Equation:GPN reconstruction}), there are exactly $(n_1+t_2)$ equations, and $(n_1+t_2)$ unknowns
    \[s, a_{1,1}, a_{1,2},\ldots, a_{1,t_1-1}, a_{2,1},\ldots, a_{2,n_1+t_2-t_1}.\]
Based on this condition, the authors of GPN scheme claimed that the system of equations (\ref{Equation:GPN reconstruction}) has a unique solution, which means that the secret can be reconstructed. GPN scheme is said to be a perfect and ideal DHSS scheme with access structure
 \[\Gamma_1=\{\mathcal{A} \subseteq \mathcal{P}: \exists \ell \in \{1,2\}~such~that~|\mathcal{A}\cap(\bigcup_{w=1}^{\ell}\mathcal{P}_{w})|\geq t_{\ell}\}.\]
\end{remark}

\section{Counterexamples}
\setcounter{equation}{0}
\renewcommand{\theequation}{B\arabic{equation}}

Let $\mathcal{P}$ be a set of $n=7$ participants. It is partitioned into $m=2$ disjoint subsets $\mathcal{P}_1=\{P_1,P_2,P_3\}, \mathcal{P}_2=\{P_4,P_5,P_6,P_7\}$ such that $n_{1}=|\mathcal{P}_{1}|=3$ and $n_{2}=|\mathcal{P}_{2}|=4$. Let $t_1=2, t_2=3$ be the threshold sequence, and the GPN scheme consists of the following share generation phase and the secret reconstruction phase.

$\bullet$ \emph{Share Generation Phase}: Let $s\in\mathbb{F}_{71}$ be a secret. The dealer chooses two random polynomials
        \[f_1(x)=a_{1,0}+a_{1,1}x, f_2(x)=a_{2,0}+a_{2,1}x+\cdots+a_{2,4}x^{4}\in \mathbb{F}_{71}[x],\]
and $7$ publicly known distinct non-zero elements $x_1, x_2, \ldots, x_7\in \mathbb{F}_{71}^{*}$ such that
                          \[f_1(0)=f_2(0)=s\in \mathbb{F}_{71},~\mathrm{and}~f(x_i)=f_2(x_i)\in \mathbb{F}_p~\mathrm{for}~i\in \{1,2,3\}.\]
The dealer computes $s_i=f(x_{i})\in \mathbb{F}_p$ for $i\in \{1,2,\ldots,7\}$, and distributes each share $s_i$ to the $i$-th participant $P_{i}$.

\begin{itemize}
    \item Secret \emph{Reconstruction Phase}: Assume that
      \[\mathcal{A}=\{i_{1},i_{2},\dots,i_{|\mathcal{A}|}\}\subseteq \mathcal{P}, i_{1}<i_{2}<\cdots<i_{|\mathcal{A}|},i_{|\mathcal{A}|}\in \mathcal{P}_\ell~\mathrm{and}~|\mathcal{A}|=t_{\ell}.\]
\end{itemize}

   \emph{Case 1}. If $\ell=1$, participants of $\mathcal{A}$ use their shares $\{s_i: i\in \mathcal{A}\}$, publicly known values $\{x_i:i\in \mathcal{A}\}$ and the Lagrange interpolating formula to reconstruct the secret
                           \[s=\sum_{i\in {\mathcal{A}}}\left(f(x_{i})\prod_{j\in {\mathcal{A}}, j\neq i}\frac{-x_{j}}{x_{i}-x_{j}}\right) \in \mathbb{F}_{71}.\]

   \emph{Case 2}. If $\ell=2$, participants of $\mathcal{A}$ use their shares $\{s_i: i\in \mathcal{A}\}$, publicly known values $\{x_i:i\in \mathcal{A}\}$ to reconstruct the secret by solving the following system of equations over $\mathbb{F}_p$.
      \begin{equation}\label{Equation: counterexmple GPN reconstruction}
       \begin{cases}
		f_1(x_i)=f_2(x_i),i\in \{1,2,3\},\\
		f_2(x_i)=s_i,i\in \mathcal{A}.\\
	    \end{cases}
      \end{equation}

As $x_1, x_2, \ldots, x_7\in \mathbb{F}_{71}^{*}$ are publicly known distinct non-zero elements, and they are randomly selected by the dealer, we will give two strategies to choose these values to attack the correctness and privacy of GPN scheme, respectively.

\begin{example}\label{Example 1} Assume that $x_1=1,x_2=2,x_3=3,x_4=4,x_5=7,x_6=8,x_7=9$, and $\mathcal{A}_1=\{P_4,P_5,P_6\}$. Clearly, $|\mathcal{A}_1\cap(\bigcup_{w=1}^{2}\mathcal{P}_{w})|=3= t_{2}$ which shows that $\mathcal{A}_1\in \Gamma_1$. We will prove that participants of $\mathcal{A}_1$ cannot reconstruct the secret.

From Equation (\ref{Equation: counterexmple GPN reconstruction}), participants of $\mathcal{A}_1$ will establish a system of equations
\[
	\begin{cases}
		(a_{2,1}-a_{1,1})x_{i}+a_{2,2}x_{i}^{2}+a_{2,3}x_{i}^{3}+a_{2,4}x_{i}^{4}=0,i \in \{1,2,3\}, \\
		s+a_{2,1}x_{i}+a_{2,2}x_{i}^{2}+a_{2,3}x_{i}^{3}+a_{2,4}x_{i}^{4}=s_{i},i \in \{4,5,6\}.
	\end{cases}\]
where $x_1=1,x_2=2,x_3=3,x_4=4,x_5=7,x_6=8$ are publicly known values, and participants of $\mathcal{A}_1$ know their own their shares $s_4, s_5, s_6$. This system of equation can be written as
\[\begin{bmatrix}		
		0 & -1 & 1 & 1 & 1 & 1 \\
		0 & -2 & 2 & 2^2 & 2^3 & 2^4\\
		0 & -3 & 3 & 3^2 & 3^3 & 3^4\\
		1 & 0 & 4 & 4^2 & 4^3 & 4^4\\
		1 & 0 & 7 & 7^2 & 7^3 & 7^4\\
		1 & 0 & 8 & 8^2 & 8^3 & 8^4 \\
	\end{bmatrix}
    \begin{bmatrix}
		s\\a_{1,1}\\a_{2,1}\\a_{2,2} \\a_{2,3}\\a_{2,4}
	\end{bmatrix}
	=\begin{bmatrix}
		0\\0\\0\\s_4\\s_5\\s_6		
	\end{bmatrix}.\]
Its general solution over $\mathbb{F}_{71}$ is
\[
	\begin{cases}
		a_{1,1}=59s+4s_4+8s_5,\\
		a_{2,1}=65s+61s_4+16s_5,\\
		a_{2,2}=60s+2s_4+9s_5,\\
		a_{2,3}=6s+57s_4+8s_5,\\
		a_{2,4}=70s+26s_4+46s_5.
	\end{cases}\]
where $s_6$ disappears, that is because
\begin{displaymath}
    \begin{aligned}
           s_6=&s+8a_{2,1}+8^{2}a_{2,2}+8^{3}a_{2,3}+8^{4}a_{2,4}\\
            =&47(s+4a_{2,1}+4^{2}a_{2,2}+4^{3}a_{2,3}+4^{4}a_{2,4})+25(s+7a_{2,1}+7^{2}a_{2,2}+7^{3}a_{2,3}+7^{4}a_{2,4})\\
            =&47s_4+25s_5\in \mathbb{F}_{71}.
    \end{aligned}
 \end{displaymath}
It is found that the secret $s$ is a freedom unknown in the general solution, which implies that participants of $\mathcal{A}_1$ cannot reconstruct the secret.
\end{example}

\begin{example}\label{Example 2} Assume that $x_1=1,x_2=2,x_3=3,x_4=7,x_5=35,x_6=9,x_7=10$, and $\mathcal{A}_2=\{P_4,P_5\}$. Clearly, $|\mathcal{A}_2\cap\mathcal{P}_{1}|=0<t_{1}$ and $|\mathcal{A}_2\cap(\bigcup_{w=1}^{2}\mathcal{P}_{w})|=2<t_{2}$. This shows that $\mathcal{A}_2\notin \Gamma_1$. We will prove that participants of $\mathcal{A}_2$ can reconstruct the secret.

From Equation (\ref{Equation: counterexmple GPN reconstruction}), participants of $\mathcal{A}_2$ will establish a system of equations
\[
	\begin{cases}
		(a_{2,1}-a_{1,1})x_{i}+a_{2,2}x_{i}^{2}+a_{2,3}x_{i}^{3}+a_{2,4}x_{i}^{4}=0,i \in \{1,2,3\}, \\
		s+a_{2,1}x_{i}+a_{2,2}x_{i}^{2}+a_{2,3}x_{i}^{3}+a_{2,4}x_{i}^{4}=s_{i},i \in \{4,5\}.
	\end{cases}\]
where $x_1=1,x_2=2,x_3=3,x_4=7,x_5=35$ are publicly known values. This system of equations can be written as
\[
	\begin{bmatrix}
		
		0 & -1 & 1 & 1 & 1 & 1 \\
		0 & -2 & 2 & 2^2 & 2^3 & 2^4\\
		0 & -3 & 3 & 3^2 & 3^3 & 3^4\\
		1 & 0 & 7 & 7^2 & 7^3 & 7^4\\
		1 & 0 & 35 & 35^2 & 35^3 & 35^4 \\
	\end{bmatrix}
    \begin{bmatrix}
		s\\a_{1,1}\\a_{2,1}\\a_{2,2} \\a_{2,3}\\a_{2,4}
	\end{bmatrix}
	=\begin{bmatrix}
		0\\0\\0\\s_4\\s_5	
	\end{bmatrix}\]
Its general solution over $\mathbb{F}_{71}$ is
\[
	\begin{cases}
		s=19s_4+53s_5,\\
        a_{1,1}=22a_{24}+38s_4+33s_5,\\
		a_{2,1}=16a_{24}+38s_4+33s_5,\\
		a_{2,2}=11a_{24},\\
		a_{2,3}=65a_{24}.\\
	\end{cases}\]
It is found that $a_{24}$ is a freedom unknown in the general solution. As a result, this system of equations have many solutions, but the secret $s=19s_4+53s_5\in \mathbb{F}_p$ is determined by participants of $\mathcal{A}_2=\{P_4,P_5\}$.
\end{example}

\begin{remark}
     GPN scheme is insecure, which is because the number of equations equal to the number of unknowns is not a sufficient and necessary condition for the system of equations to have a unique solution, such as the Counterexample \ref{Example 1}. Besides, even if a system of equations has many solutions, one of the unknowns (the secret) may be uniquely determined, such as the Counterexample \ref{Example 2}.
\end{remark}

\end{document}